\documentclass[11pt,a4paper]{amsart}
\usepackage[text={4.5in,7in},centering]{geometry}
\makeatletter

\g@addto@macro{\endabstract}{\@setabstract}
\newcommand{\authorfootnotes}{\renewcommand\thefootnote{\@fnsymbol\c@footnote}}%
\makeatother

\usepackage[dvips]{graphicx}
\usepackage{amsmath}
\usepackage{color}
\usepackage{subfigure}
\usepackage{pgf,tikz}
\usetikzlibrary{arrows}
\usepackage{tikz}

\newcommand{\inD}{\mathcal{D}}

\theoremstyle{plain}
\newtheorem{theorem}{Theorem}
\newtheorem{lemma}[theorem]{Lemma}

\newtheorem{problem}{Problem}
\newtheorem{definition}[theorem]{Definition}
\newtheorem{proposition}[theorem]{Proposition}

\begin{document}
\begin{center}
  \LARGE 
  Point Set Isolation Using Unit Disks is $\mathsf{NP}$-complete \par \bigskip

  \normalsize
  \authorfootnotes
Rainer Penninger \textsuperscript{1} and Ivo Vigan\footnote{Research supported by NSF grant 1017539}\textsuperscript{2}
 \par \bigskip

  \textsuperscript{1}Dept. of Computer Science I, University of Bonn \par\smallskip
  \textsuperscript{2}Department of Computer Science, City University of New York,\\ The Graduate Center, New York\par \bigskip

\end{center}

\begin{abstract} 
We consider the situation where one is given a set $S$ of points in the plane and a collection $\mathcal{D}$ of unit disks embedded in the plane. We show that finding a minimum cardinality subset of $\mathcal{D}$ such that any path  between any two points in $S$ is intersected by at least one disk is  $\mathsf{NP}$-complete.  This settles an open problem raised in \cite{discSep}. Using a similar reduction, we show that finding a minimum cardinality $\mathcal{D}' \subseteq \mathcal{D}$ such that $\mathbb{R}^2 \setminus \bigcup ( \mathcal{D} \setminus \mathcal{D}')$ consists of a single connected region is also $\mathsf{NP}$-complete. Lastly, we show that the Multiterminal Cut Problem remains $\mathsf{NP}$-complete when restricted to unit disk graphs. 
\end{abstract}

%\maketitle

\section{Introduction and Main Results}

In this note we show that the (decision version of the) Point Set Isolation Problem defined below in Problem 1 is $\mathsf{NP}$-complete. This problem was introduced in \cite{discSep} where a polynomial-time constant-factor approximation algorithm was presented, but the problem complexity was stated as an open problem. As a motivation for studying this problem, in \cite{discSep},  its relevance to trap coverage in sensor networks is mentioned, where one wants to detect certain spacial transitions among the observed objects  (see for example \cite{bar1}).  

\begin{problem}[Point Set Isolation Problem \cite{discSep}]
Given a set $S$ of $k$ points in the plane and a collection $\mathcal{D}$ of $n$ unit disks embedded in the plane, find a minimum subset $\mathcal{D'} \subseteq \mathcal{D}$, such that every path between two points in $S$ is intersected by at least one disk in $\mathcal{D'}$.
\label{diskP}
\end{problem}

Using similar reductions as for showing hardness of the Point Set Isolation Problem in Section \ref{secP1}, we show that the All-Cells-Connection Problem (Problem \ref{allCellP}) and the Unit Disk Multiterminal Cut Problem (Problem \ref{mutiC}) are $\mathsf{NP}$-complete in Sections \ref{allCellSec} and  \ref{MCSec} respectively.

\begin{problem}[All-Cells-Connection Problem]
Given a set $\mathcal{D}$ of unit disks embedded in the plane, find a minimum cardinality subset $\mathcal{D}'$ such that $\mathbb{R}^2 \setminus \bigcup ( \mathcal{D} \setminus \mathcal{D}')$ consists of a single connected region.
\label{allCellP}
\end{problem}

In \cite{alt} the All-Cells-Connection Problem has been shown to be $\mathsf{APX}$-hard for arrangements induced by line segments. \\ Following the reasoning of \cite{alt}, looking at the Point Set Isolation Problem as a trap cover problem, it asks for the smallest number of sensors which need to be turned on in order to prevent any point in $S$ of leaving its cell. The All-Cells-Connection Problem then asks for the minimum number of sensors to be turned off so that any point of $S$ can move freely between any previously existing cells of the sensor network.

\begin{problem}[Unit Disk Multiterminal Cut Problem \cite{multicut}]
Given a unit disk graph $G=(V,E)$ and a set $S \subseteq V$ of $k$ terminals, find a minimum cardinality set $E' \subseteq E$ such that in $G' = (V, E \setminus E')$ there is no path between any two nodes in $S$.
\label{mutiC}
\end{problem}

The general idea for showing hardness of the three problems is to take an instance of a hard problem on planar graphs and embed it on an integer grid using straight line segments. Each embedded edge then gets replaced by an edge gadget consisting of certain unit disk arrangements. The dimension of each edge gadget is chosen such that no two unit disks of different edge gadgets intersect. Furthermore, we replace each embedded vertex $v$ by a vertex gadget which consists of a cycle of unit disks which is circularly arranged around $v$. Each edge gadget of edges incident to $v$ will intersect a small number of disks contained in the vertex gadget. The main task of the reduction is to choose the radius of the disks and the dimension of the gadgets such that the edge gadgets are large enough to contain the problem-specific disk arrangements and small enough to ensure that non-incident edge gadgets are disjoint. %We would like to note that this technique might be used to obtain other hardness results in unit disk settings. \\

\section{Hardness of  Point Set Isolation  }
\label{secP1}

In this section we prove the following theorem by reducing an instance of the Planar Subdivision Problem (Problem \ref{planarSubDev}) to it.

\begin{theorem}
The Point Set Isolation Problem is  $\mathsf{NP}$-complete if $k$ is not fixed.
\label{mainThm}
\end{theorem}

\begin{problem}[Planar Subdivision Problem]
Given a simple planar graph $G=(V,E)$ embedded in the plane and a set $S$ of $k$ points properly contained in the faces of $G$ with no face containing more than one point, find a minimum cardinality set $E' \subseteq E$ such that in the embedding of the reduced graph $G'=(V',E')$, with $V' = \{v \in e : e \in E'\}$, no two points are contained in the same face.
\label{planarSubDev}
\end{problem}

\begin{proposition}
The Planar Subdivision Problem is $\mathsf{NP}$-complete if $k$ is not fixed, even on connected graphs.
\label{intProp}
\end{proposition}
\begin{proof}
Given an instance $I_1= (G_1, S_1)$ of the Planar unweighted Multiterminal Cut Problem \cite{multicut}, with $G_1 = (V_1, E_1)$ we embed $G_1$ in the plane and build an instance $I_2 = (G_2, S_2)$ of the Planar Subdivision Problem with $G_2 = (V_2, E_2)$. This is done by letting $G_2$ be the geometric dual graph of the embedded graph $G_1$ and further subdivide each dual edge $\{u,v\}$ into $\{u,x\}$ and $\{x,v\} $, thus-by ensuring that $G_2$ is simple. We embed $G_2$ in the plane and build the set $S_2$ by putting a point into the interior of each face of $G_2$ whose dual vertex is in $S_1$. Since for any dual edge $\{u,v\}$, taking only $\{u,x\}$ or only $\{x,v\}$ will not change the partitioning of the plane, an optimal solution $OPT$ for $I_2$ only contains pairs of subdivision edges. Letting $OPT_{\cup}$ denote the set of edges obtained from $OPT$ by merging any subdivision edges $\{u,x\} \in OPT$ and $\{x,v\} \in OPT$ into $\{u,v\}$ and let $OPT^*_{\cup} \subseteq E_1$ denote the duals of the edges in $OPT_{\cup}$. We claim that $OPT$ is an optimal solution for $I_2$ if and only if $OPT^*_{\cup}$ is an optimal solution for $I_1$. 
To see this, let $E' \subseteq E_1$ and let $E'^{*}$ be the corresponding dual edges. If two vertices $u,v \in S_1$ are connected by a path $(u, v_1, \ldots, v_l, v )$ in $G_1'  = (V_1, E \setminus E')$ then, the sequence $u^*, v^*_1, \ldots, v^*_l, v^*$ of adjacent faces in $G_2$ is merged to one face in $G_2' = (V_2, E'^{*})$ and thus the points $p_u, p_v \in S_2$ corresponding to $u,v$ are contained in the same face. By the same argument it follows that if  $p_u, p_v \in S_2$ are contained in the same face, then $u,v \in S_1$ are connected by a path. Thus $I_2$ has a solution of size $2M$ if and only if $I_1$ has a solution of size $M$. Since the Planar unweighted Multiterminal Cut Problem is $\mathsf{NP}$-complete  on connected graphs and the dual of a connected graph is connected, the Planar Subdivision Problem is $\mathsf{NP}$-complete even on connected graphs.
 \end{proof}

Since an instance of the Planar Subdivision Problem can be any jordan arc embeddeding of a graph in the plane, we need to argue that replacing the embedding by a straight line grid embedding does not change the solution of the  Planar Subdivision Problem.

\begin{lemma}
For any jordan arc embedding of a Planar Subdivision instance $(G,S)$ with $G = (V, E)$ being a connected graph on $n$ vertices, there exists a straight line embedding of $G$ on an $n \times n$ integer grid, such that every solution in the original embedding is a solution in the grid embedding and vice versa.
\label{embedLem}
\end{lemma}
\begin{proof}
As shown in \cite{gridEmd} every planar graph on $n$ vertices can be drawn crossing free on an $n \times n$ grid using straight line segments in time $O(n)$. Furthermore, it holds that every maximal plane graph on at least four vertices is three-connected and every three-connected graph has a unique embedding by Whitneys Theorem (modulo the choice of the outer face). While Whitneys Theorem holds for an even stronger notions of equivalence, we say that two embeddings of a connected graph are \emph{equivalent} if for each vertex all incident edges have the same circular clockwise order in both embeddings. Making $G$ maximally plane, embedding it on an $n \times n$ grid and removing the additional edges results in an equivalent grid embedding of $G$ compared to its original embedding. The only thing left to show is that two equivalent embeddings have the same set of solutions. To see this, note that the set of circular orders around the vertices uniquely defines the facial walks in the embeddings. Therefore, two equivalent embeddings have the same geometric dual graph. This holds, since two adjacent faces have some common edge in their corresponding facial walks and thus the dual of the faces are connected by an edge. Since equivalent embeddings have the same facial walks, their dual graphs are isomorphic. Thus if there are two points which are not separated in one embedding, there exists a path between their corresponding faces in the dual graph, but due to isomorphism of the dual graphs for equivalent embeddings, the two points are not separated in all equivalent embeddings.
\end{proof}

In order to prove Theorem \ref{mainThm} we reduce an instance $I_2 =(G_2, S_2)$ of the Planar Subdivision Problem, with $G_2$ being a connected embedded graph, in polynomial time to an instance $I_1 = ( {\inD}, S_1)$ of the Point Set Isolation Problem. We do this by first transforming the embedding of $G_2$ to an  \emph{equivalent} straight line embedding on an  $n \times n$ integer grid as argued in Lemma \ref{embedLem}.

 We then replace each edge in the embedding by an edge gadget defined in Definition \ref{edgeGdEf}. An edge gadget as depicted in Figure~\ref{edgeG} consists of a path of unit disks constructed in such a way that every edge gadget contains the same amount of unit disks, regardless of the length of the embedded edge. Furthermore, the dimensions of each edge gadget is chosen such that no two unit disks of different edge gadgets intersect. Having replaced each edge by an edge gadget, we replace each vertex $v$ by a vertex gadget defined in Definition \ref{vertexGdEf}. A vertex gadget for $v$ consists of a cycle of unit disks which is circularly arranged around $v$. Each edge gadget of edges incident to $v$ will intersect a small number of disks contained in the vertex gadget. If we denote the collection of all disks contained in the vertex- and edge gadgets by $ {\inD}$, then each face in $G_2$ has a corresponding connected region in $\mathbb{R}^2 \setminus  {\bigcup \inD}$. We then place for each $s_2 \in S_2$ a point $s_1$ into the region in $\mathbb{R}^2 \setminus  {\bigcup \inD}$ corresponding to the face where $s_2$ was placed in $G_2$ and add $s_1$ to $S_1$. The main task of the reduction is to choose the radius of the disks and the dimension of the gadgets such that every edge gadget consists of the same amount of disks and that all edge gadgets are disjoint. Thus, on such an instance of the Point Set Isolation problem removing any disk from an edge gadget merges two regions in $\mathbb{R}^2 \setminus  {\bigcup \inD}$. As argued in Lemma \ref{redLem}, given a solution of $I_1$, one can retrieve a solution for $I_2$ in polynomial time by removing all those edges in $G_2$ whose corresponding edge gadgets were removed in the solution of $I_1$. 
In order to retrieve the maximal number of edges removed from $G_2$ from the number of disks removed in the solution of $I_1$ we choose the number of unit disks in a single edge gadget to be larger than the number of unit disks in all the vertex gadgets combined. This then allows to retrieve the maximal number of edges removed from $G_2$ from the number of disks removed in the solution of $I_1$.

%--------

\definecolor{xdxdff}{rgb}{0.49,0.49,1}
\definecolor{qqwuqq}{rgb}{0,0.39,0}
\definecolor{qqqqff}{rgb}{0,0,1}
\definecolor{uququq}{rgb}{0.25,0.25,0.25}

\definecolor{black}{rgb}{0,0,1}
\definecolor{black}{rgb}{0.25,0.25,0.25}
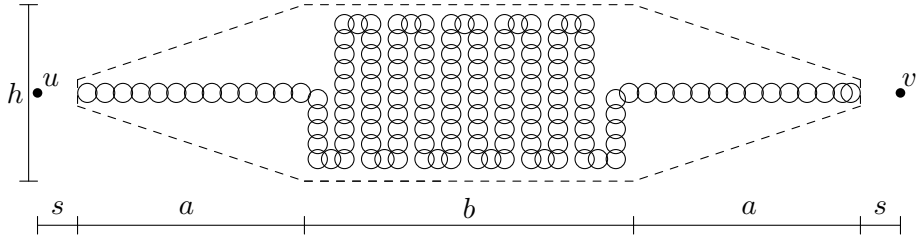
\begin{figure}
\center
\begin{tikzpicture}[line cap=round,line join=round,>=triangle 45,x=1.3cm,y=1.3cm, scale=0.9]
\clip(-0.38,-2) rectangle (10,2);
%\draw [shift={(0,0)},color=qqwuqq,fill=qqwuqq,fill opacity=0.1] (0,0) -- (0:0.38) arc (0:18.43:0.38) -- cycle;

%---new--------

\draw [dashed ](3,1)-- (6.7,1);
\draw [dashed ](3,-1)-- (6.7,-1);
\draw [dashed ](3,-1)-- (4.54,-1);
\draw [dashed ](0.45,0.15)-- (0.45,-0.15);
\draw [dashed ] (0.45,0.15)-- (3,1);
\draw [dashed ] (0.45,-0.15)-- (3,-1);
\draw [dashed ] (9.25,0.15)-- (9.25,-0.15);
\draw [dashed ] (9.25,0.15)-- (6.7,1);
\draw [dashed ] (9.25,-0.15)-- (6.7,-1);

%---old--------
%\draw  [dotted] (0,0)-- (4.5,1.5);

%\draw [dashed](3,0.92)-- (6.7,0.92);
%\draw [dashed] (3,-0.92)-- (6.7,-0.92);
%\draw  [dashed](0.45,0.15)-- (0.45,-0.15);
%\draw [dashed](3,0.15)-- (3,0.92);
%\draw [dashed](0.45,0.15)-- (3,0.15);
%\draw [dashed](0.45,-0.15)-- (3,-0.15);
%\draw [dashed](3,-0.15)-- (3,-0.92);
%\draw [dashed](9.25,0.15)-- (9.25,-0.15);
%\draw  [dashed]  (6.7,0.15)--(6.7,0.92);
%\draw [dashed]   (6.7,0.15)-- (9.25,0.15);
%\draw  [dashed]  (6.7,-0.15)-- (9.25,-0.15);
%\draw  [dashed] (6.7,-0.15)--(6.7,-0.92);
%---------

%vertices
 \fill (0 ,0) circle (2pt);
\draw (0.15, 0.15) node {$u$};

\fill (9.7 ,0) circle (2pt);
\draw (9.8,0.15) node {$v$};

%scale
\draw (0,-1.5)-- (9.7,-1.5);

\draw (0.225,-1.3) node {$s$};
\draw (0,-1.40)-- (0,-1.6);
\draw (1.675,-1.3) node {$a$};
\draw (0.45,-1.40)-- (0.45,-1.6);
\draw (3,-1.40)-- (3,-1.6);
\draw (4.85,-1.3) node {$b$};
\draw (6.7,-1.4)-- (6.7,-1.6);
\draw (8,-1.3) node {$a$};
\draw (9.25,-1.4)-- (9.25,-1.6);
\draw (9.475,-1.3) node {$s$};
\draw (9.7,-1.4)-- (9.7,-1.6);

\draw (-0.1,-1)-- (-0.1,1);
\draw (-0.2,1)-- (0,1);
\draw (-0.2,-1)-- (0,-1);
\draw (-0.25, 0) node {$h$};

%\draw (9.6,-0.15)-- (9.6,0.15);
%\draw (9.5,0.15)-- (9.7,0.15);
%\draw (9.5,-0.15)-- (9.7,-0.15);
%\draw (10, 0) node {$2r$};

\begin{scriptsize}
%\draw[color=qqwuqq] (0.2,0.22) node {$\leq \alpha$};
\end{scriptsize}

\foreach \x in {0.56, 0.76,...,3.12} {
      \draw (\x ,0) circle (4pt);
}
\foreach \x in {6.65, 6.85,...,9.2} {
      \draw (\x ,0) circle (4pt);
}

 \draw (9.14 ,0) circle (4pt);

\foreach \x in {3.6,4.2,...,6.5} {
      \draw (\x ,0.78) circle (4pt);
}
\foreach \x in {3.3,3.9,...,6.5} {
      \draw (\x ,-0.75) circle (4pt);
}

%\foreach \x in {3.15,3.45,...,6.5} {
\foreach \x in {3.45,3.75,...,6.2} {
\foreach \y in {-0.75, -0.58,...,0.8} {
      \draw (\x ,\y) circle (4pt);
}
}

\foreach \y in {-0.75, -0.58,...,0} {
      \draw (3.15 ,\y) circle (4pt);
}
\foreach \y in {-0.75, -0.58,...,0} {
      \draw (6.5 ,\y) circle (4pt);
}

\end{tikzpicture}
\caption{The edge gadget for the edge $\{u,v\}$. }
\label{edgeG}
\end{figure}

%--------

\begin{definition}
An edge gadget for an embedded edge $e = \{u,v\}$ of length $2s+2a + b$ is a path of unit disks which can be thought of as being placed in an elongated octagon of height $h$ which has a cabin of length $b$ as shown Figure~\ref{edgeG}. Every edge gadget consists of a path of $C_E$ many unit disks which is a straight path in the two hallways of length $a$ and an up-down path in the middle cabin. While $C_E$, $a$, $h$ and $s$ are constant for all edge gadgets, $b$ may vary from $1-(2s+2a)$ to $\sqrt{2}n-(2s+2a)$, depending on the length of the embedded edge $e$.
\label{edgeGdEf}
\end{definition}

We set the number $C_E$ of unit disks contained in each edge gadget to $ \lceil \frac{\sqrt{2}n-2s}{2r} \rceil$ since this amounts to the number of disks of radius $r$ needed to represent the longest edge in an $n \times n$ grid embedding as a straight line chain of unit disks. Since we force each edge gadget to contain an equal amount of disks, we need to place $C_E - a/r$ disks into the cabin of any edge gadget.  Arranging the disks in cabins of edges of length $<  \sqrt{2}n$ as an up-down path as described in Definition \ref{edgeGdEf} allows us to put a path consisting of up to $\left\lfloor\frac{h}{2r}\right\rfloor\cdot\left\lfloor\frac{b}{2r}-1\right\rfloor$ many disks into the cabin.\\
Truncating the edge gadgets by a sufficiently large distance $s$ from both of its endpoints ensures that no two unit disks contained in the hallways of adjacent edge gadgets intersect each other. According to Lemma \ref{smallAngl} the smallest angle between two edges is greater than $2\arctan 1/(6n^{2})$. Thus, setting $s$ to $r\frac{1-\sin{(2\arctan (1/6n^2)/2)}}{\sin{(2\arctan (1/6n^2)/2)}} = r(\sqrt{36n^4+1}-1)$ ensures disjointness of the disks in adjacent hallways.

\begin{definition}

A vertex gadget for an embedded vertex $v$ consists of a cycle of $C_V =  \lceil\pi s/r\rceil$ many unit disks of radius $r$ which are arranged on a circle of radius $s$ centered at $v$.
\label{vertexGdEf}
\end{definition}

In order for the construction to work we need to chose the radius $r$ of the unit disks and the height $h$ of the edge gadgets such that the following constraints hold:
\begin{eqnarray}
2(r+s) &<& 1\\
r+s+h/2 &<&  (2n^2-2n+1)^{-\frac{1}{2}}\\
2\frac{h}{2} &<& (2n^2-2n+1)^{-\frac{1}{2}}\\
h/2 &<& \frac{s+a}{6n^2}\\
\left\lceil\frac{\sqrt{2}n-2s}{2r}\right\rceil  - 2\left \lceil \frac{a}{2r} \right \rceil &\leq& \left\lfloor\frac{h}{2r}\right\rfloor\cdot\left\lfloor\frac{1-2(s+a)}{2r}-1\right\rfloor
\end{eqnarray}

with $a$ being a parameter for the gadget subdivision and $s$ being fixed to $ r\sqrt{36n^4+1}-1$ as described above.
Note that any point inside a disk in the edge gadget for edge $e$ has distance at most $h/2$ to $e$, and a point inside a disk in the vertex gadget for vertex $v$ has distance at most $r+s$ to $v$.
Thus, the first constraint ensures that no two vertex gadgets intersect
and, using Lemma~\ref{minDist}, the second constraint assures that no edge gadget intersects any vertex gadget other than the ones at its two endpoints.
Analogously, inequality 3 implies that no edge gadget intersects another edge gadget if the corresponding edges do not share a common vertex.
Inequality $4$ ensures that no disk placed in the cabin of an edge gadget intersects any disk in any incident edge gadget. The right hand side of this constraint follows from the fact that according to Lemma \ref{smallAngl} the small angle between two incident edges is larger  than $2\arctan 1/(6n^{2})$. Since the cabin of the edge gadget starts at a distance of $s+a$ from the incident vertex, restricting that the cabin extends less than $\frac{s+a}{6n^2}$ from the embedded edge ensures disjointness of disks contained in cabins of different edge gadgets.
The fifth constraint assures that the cabin of every edge gadget is big enough so that the whole gadget contains a path of $C_E$ many disks.

If we set $a=1/4$ it follows that $b$ is at least $1/2-2s$ in every edge gadget. Doing some tedious calculations shows that a radius of $r = \frac{1}{40n^4}$ and a height of $h = \frac{1}{12n^2}$ satisfies the five constraints simultaneously for all $n \geq 2$.

Having $r$ and $h$ fixed, it follows that for any edge $e= \{u,v\}$ in the graph of an instance of Problem \ref{planarSubDev}, the first and last disks in the edge gadget of $e$ intersect between one and three disks from the vertex gadgets of $u$ and $v$ respectively. %This follows since the center of the disks of the vertex gadgets are a distance $s$ away from its vertex, while the center of the closest disk in any edge gadget is $s+r$ away from the same vertex. Thus, due to rounding there can be up to three disks in the vertex gadget intersecting the edge gadget. \\
% By Lemma \ref{smallAngl} and the choice of $h$ and $s$ it follows that no disk of the edge gadget of $e$ intersects any disk in any other edge gadget. Furthermore, due to the first constraint no disk in the edge gadget of $e$ intersects any disks of any vertex gadget other than the ones for $u$ and $v$. \\

Plugging the calculated values for $r$ and $h$ into $C_E$ and $C_V$ yields that an edge gadget consists of
$\left\lceil(20\sqrt{2}n^5-\sqrt{36n^4+1})+1\right\rceil$
many disks and a vertex gadget consists of
$\left\lceil\pi(\sqrt{36n^4+1}-1)\right\rceil$
disks and we can conclude that the above construction can be done in polynomial time.

\begin{lemma}
In an $n \times n$ grid, the minimum distance between any line $l$ through two grid points and any grid point not on $l$ is $(2n^2-2n+1)^{-\frac{1}{2}}$. 
\label{minDist}
\end{lemma}
\begin{proof}
Wlog we fix one point on $l$ to $(0,0)$. Denoting the second point on $l$ by $(a,b)$, we get a line equation of $bx-ay=0$. Thus, the distance from a point $(c,d)$ to $l$ is $\frac{|bc-ad|}{\sqrt{b^2+a^2}}$. Furthermore, we can assume that $gcd(a,b)=1$ since otherwise we can divide both coordinates by $gcd(a,b)$. 
Thus, setting $a=n$ and $b=n-1$ maximizes $a^2+b^2$, given $gcd(a,b)=1$, and the minimum non-zero distance is thus at least $(2n^2-2n+1)^{-\frac{1}{2}}$. Furthermore, observing that $|bc-ad|\geq 1$ yields that the minimum value is achieved at point $(c,d)=(1,1)$. 
\end{proof}

%Since the minimum distance between two non-intersecting line segments occurs at one of their endpoints, Lemma~\ref{minDist} provides a lower bound for the minimum distance between two disjoint embedded edges.

\begin{lemma}
In an $n \times n$ grid, for any grid point $p$ the minimum angle between any two distinct lines, each going through $p$ and at least one other grid point, is larger than $2\arctan 1/(6n^{2})$.
\label{smallAngl}
\end{lemma}
\begin{proof}
Let $g,h$ denote two lines through $p$ and $(a,b)$, $(c,d)$ respectively with minimum angle, and let the slope of $g$ be larger than the slope of $h$, Thus, having $b/a > d/c$. Wlog $p = (0,0)$ and due to symmetry we can restrict $g$ and $h$ to be contained in the lower right triangle $\{(i,j) | 0 \leq j \leq i \leq n\}$. 
Now  due to monotonicity of  $\arctan$ it holds that $\arctan b/a - \arctan d/c = \arctan \frac{bc - ad}{1-(bd)/(ac)} \geq  \arctan \frac{b/a-
d/c}{2}  \geq  \arctan  \frac{1}{2n^{2}}$. The last inequality holds since  all coordinates are integers, thus $bc - ad \geq 1$ and therefore $b/a - d/c = \frac{bc - ad}{ac} \geq 1/n^2$.  The Lemma then follows using the fact that $\arctan(x) > 2\arctan(x/3)$ holds for all $0 < x < \sqrt{3}$.
\end{proof}

In order to proof Theorem \ref{mainThm} we need to show how to retrieve a solution for an instance of Problem \ref{planarSubDev} from the solution of the disk arrangement built using the construction outlined above. 

\begin{lemma}
An instance $I_2=(G, S)$ of Problem \ref{planarSubDev} has a solution of size at most $B$ if and only if $I_1$ of the Point Set Isolation Problem has a solution of size at most $C_E(B+1)-1$, where $I_1$ is built out of $I_2$ using the construction described above.
\label{redLem}
\end{lemma}
\begin{proof}
Given an instance $I_2=(G, S)$ we embed $G = (V,E)$ crossing free into an $n \times n$ grid and replace each edge in $E$ by an edge gadget as described in Definition \ref{edgeGdEf} and each vertex by a vertex gadget as described in Definition \ref{vertexGdEf}. If we denote the collection of all disks contained in vertex- and edge gadgets by $ {\inD}$, we place each $s \in S$ in the region in $\mathbb{R}^2 \setminus  {\bigcup \inD}$ corresponding to the face where $s$ was placed in the embedding of $G$. Note that by construction such a face exists and that removing any disk from any edge gadget merges two regions in $\mathbb{R}^2 \setminus  {\bigcup \inD}$ into one. On the other hand, removing disks from a vertex gadget might not merge any regions, since they can still be separated by the edge gadgets which are adjacent to the vertex gadget. Thus, a valid solution may use all or no disks of any vertex gadget. It thus, follows that $I_2$ has a solution of size $k_2$ if $I_1$ has a solution of size $k_1$ with $k_2 C_E \leq k_1 \leq k_2 C_E+nC_v < (k_2+1) C_E$, where the last inequality follows from the fact that $nC_V < C_E$. On the other hand if $I_1$ has a solution consisting of $k_1$ disks, then $I_2$ has a solution of size $k_2$ exactly when $k_1/C_E-1< k_2 \leq k_1/ C_E$ which again implies $k_2C_E \leq k_1 < (k_2+1) C_E$.
\end{proof}

The last thing to show is that the decision version of the Point Set Isolation problem is contained in  $\mathsf{NP}$, where the decision version asks if a solution to an instance of the Point Set Isolation problem exists whose size is at most $B$. To show containment in $\mathsf{NP}$, it suffices to show that  a solution to the decisions version can be verified in polynomial time. To see this, let $\mathcal{D}'$ be a possible solution to the problem. Testing whether $|\mathcal{D}'| \leq B$ is trivial. In order to test whether $\mathcal{D}'$ indeed separates all points, we construct the embedded unit disk graph of $\mathcal{D}'$ and add a vertex at every edge crossing. Thus we obtain a line segment arrangement $\mathcal{A} = \mathcal{A}(\mathcal{D}')$. We reject $\mathcal{D}'$ as a solution, if and only if two point location queries \cite{pointLoc} for distinct points in $S$ return the same face of $\mathcal{A}$. Since every line segment of $\mathcal{A}$ is completely contained in some disk of $\mathcal{D}'$, the points in $S$ are separated in $\mathcal{D}'$ if they are separated $\mathcal{A}$. On the other hand, if a point $s$ is contained in a region bounded by some disks $\mathcal{D}'_s$ which separates $s$ from all other points in $S$, then $\mathcal{D}'_s$ form a cycle in the unit disk graph and thus there exists a face in $\mathcal{A}$ containing only $s$.

%We use the polynomial time algorithm of \cite{barrier} which, for two points $s,t$, calculates the  \emph{resilience} of a disk arrangement  w.r.t. $s$ and $t$, i.e. the minimum number of disks intersected by any path connecting $s$ and $t$, where the same disk might be counted multiple times. If the algorithm returns a non-zero value for any two distinct points in $S$ then we conclude that $\mathcal{D}'$ is a solution, otherwise we reject. 

\section{Hardness of All-Cells-Connection}
\label{allCellSec}
\begin{theorem}
The All-Cells-Connection Problem is $\mathsf{NP}$-complete.
\label{mainThm3}
\end{theorem}

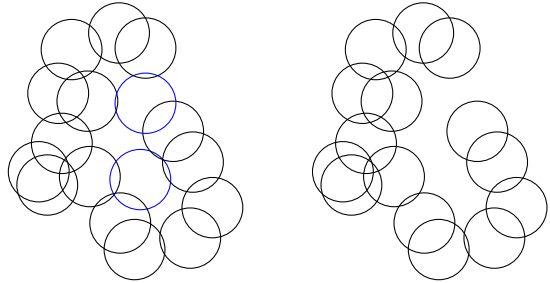
\begin{figure}
\center
\subfigure{
\definecolor{xdxdff}{rgb}{0.66,0.66,0.66}
\definecolor{qqqqff}{rgb}{0.33,0.33,0.33}
\begin{tikzpicture}[line cap=round,line join=round,>=triangle 45,x=.4cm,y=.4cm]
\clip(2,-0.2) rectangle (10.4,9.4);
\draw(4.12,4.34) circle (0.4cm);
\draw(5.04,3.24) circle (0.4cm);
\draw(3.64,2.96) circle (0.4cm);
\draw [color=blue] (6.7,3.14) circle (0.4cm);
\draw(7.78,4.74) circle (0.4cm);
\draw(4.96,5.74) circle (0.4cm);
\draw [color=blue] (6.87,5.67) circle (0.4cm);

\draw(3.36,3.4) circle (0.4cm);
\draw(6.04,1.7) circle (0.4cm);
\draw(4,6) circle (0.4cm);
\draw(4.44,7.45) circle (0.4cm);
\draw(6,8) circle (0.4cm);
\draw(6.87,7.5) circle (0.4cm);
\draw(6.51,0.82) circle (0.4cm);
\draw(8.34,1.2) circle (0.4cm);
\draw(9.07,2.21) circle (0.4cm);
\draw(8.43,3.72) circle (0.4cm);	
\end{tikzpicture}
}
\subfigure{
\definecolor{xdxdff}{rgb}{0.66,0.66,0.66}
\definecolor{qqqqff}{rgb}{0.33,0.33,0.33}
\begin{tikzpicture}[line cap=round,line join=round,>=triangle 45,x=.4cm,y=.4cm]
\clip(2,-0.2) rectangle (10.4,9.4);
\draw(4.12,4.34) circle (0.4cm);
\draw(5.04,3.24) circle (0.4cm);
\draw(3.64,2.96) circle (0.4cm);
\draw(7.78,4.74) circle (0.4cm);
\draw(4.96,5.74) circle (0.4cm);
\draw(3.36,3.4) circle (0.4cm);
\draw(6.04,1.7) circle (0.4cm);
\draw(4,6) circle (0.4cm);
\draw(4.44,7.45) circle (0.4cm);
\draw(6,8) circle (0.4cm);
\draw(6.87,7.5) circle (0.4cm);
\draw(6.51,0.82) circle (0.4cm);
\draw(8.34,1.2) circle (0.4cm);
\draw(9.07,2.21) circle (0.4cm);
\draw(8.43,3.72) circle (0.4cm);	
\end{tikzpicture}
}
\label{allCellEx}
\caption{An instance $\mathcal{D}$ of the All-Cells-Connection Problem with its optimal solution $\mathcal{D}'$ depicted in blue on the left and the resulting arrangement on the right where $\mathbb{R}^2 \setminus \bigcup ( \mathcal{D} \setminus \mathcal{D}')$ consists of a single connected region.}
\end{figure}

\begin{proof}
In order to prove Theorem \ref{mainThm3}, we are going to use a restricted version of the Feedback Vertex Set Problem (FVS). Given a graph, FVS asks for the minimum cardinality subset of vertices to be removed such that the remaining graph is acyclic. In \cite{fvsPhd} (see also \cite{fvsPaper}), it is shown that FVS is $\mathsf{NP}$-complete in undirected planar graphs with maximum degree $4$. Given such an instance $G=(V,E)$, we embed it into a grid and replace each edge with the edge gadget\footnote{The fact that all edge gadgets contain the same amount of disks is irrelevant for the reduction. The only relevant property is that the disks in each edge gadget form a simple path.} of Definition \ref{edgeGdEf}.  For each vertex $v$, we build a vertex gadget as shown in Figure \ref{fvsVG}. Note that all edge gadgets of edges incident to $v$ end at a distance of $s = r(\sqrt{36n^4+1}-1)$ from $v$. This fact together with Lemma \ref{smallAngl} implies that there is enough space for up to four disjoint concentric circular unit disk paths connecting the incident edge gadgets to the (red) center disk. We connect each of those paths to the center disk in such a way that removing the center disk merges all of the at most four faces incident to $v$ into one. The resulting unit disk arrangement  $\mathcal{D}$ consists of $|E|$ simple unit disk paths. It is thus clear that $G$ has a FVS of size $k$ if and only if there is a set $\mathcal{D}' \subseteq \mathcal{D}$ of size $k$ such that $\mathbb{R}^2 \setminus \bigcup ( \mathcal{D} \setminus \mathcal{D}')$ consists of a single connected region.
\end{proof}

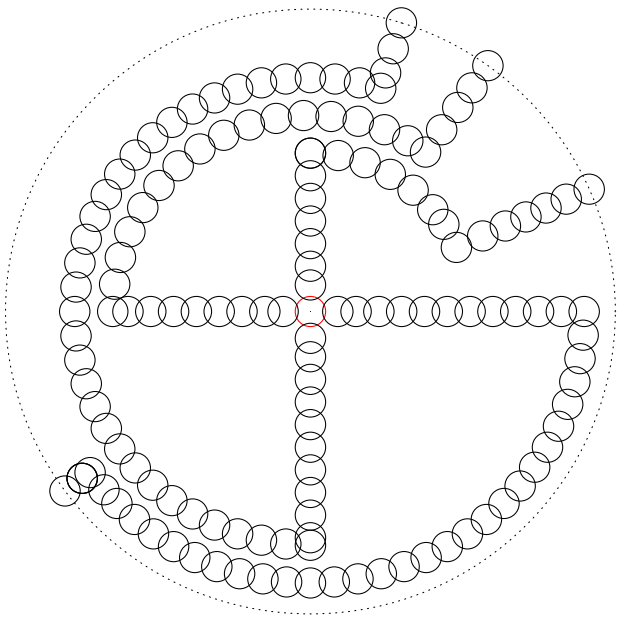
\begin{figure}
\center
\begin{tikzpicture}[line cap=round,line join=round,>=triangle 45,x=1.0cm,y=1.0cm, scale=0.2]
\clip(-21,-21) rectangle (21, 21);
\draw[dotted](0,0) circle (20.05cm);
\draw [red] (0,0) circle (1cm);
\draw(5.98,19.14) circle (1cm);
\draw(0,1.76) circle (1cm);
\draw(1.8,0) circle (1cm);
\draw(0,-1.78) circle (1cm);
\draw(-1.8,0) circle (1cm);
\draw(11.67,16.3) circle (1cm);
\draw(18.33,8.11) circle (1cm);
\draw(-16.14,-11.9) circle (1cm);
%\draw[dashed](0,0) circle (15.5cm);
%\draw[dashed](0,0) circle (13cm);
%\draw[dashed](0,0) circle (10.5cm);
%\draw[dashed](0,0) circle (8cm);
\draw(-15.01,-11.06) circle (1cm);
%\draw(0,0) circle (18cm);
\draw(-15.01,-11.06) circle (1cm);
\draw(-15.01,-11.06) circle (1cm);
\draw(16.82,7.45) circle (1cm);
\draw(15.49,6.86) circle (1cm);
\draw(14.17,6.27) circle (1cm);
\draw(-14.49,-10.68) circle (1cm);
\draw(-13.58,-11.82) circle (1cm);
\draw(5.44,17.42) circle (1cm);
\draw(4.94,15.82) circle (1cm);
\draw(4.62,14.8) circle (1cm);
\draw(3.23,15.16) circle (1cm);
\draw(10.62,14.83) circle (1cm);
\draw(9.69,13.53) circle (1cm);
\draw(8.63,12.05) circle (1cm);
\draw(7.57,10.57) circle (1cm);
\draw(6.4,11.32) circle (1cm);
\draw(12.82,5.67) circle (1cm);
\draw(11.33,5.01) circle (1cm);
\draw(9.6,4.25) circle (1cm);
\draw(8.77,5.78) circle (1cm);
\begin{scriptsize}
%\fill [color=black] (0,0) circle (1.5pt);
%\draw[color=black] (0.47,0.75) node {$A$};
\fill [color=black] (20.05,0) circle (1.5pt);
%\draw[color=black] (20.42,0.75) node {$B$};
\fill [color=black] (5.98,19.14) circle (1.5pt);
%\draw[color=black] (6.41,19.89) node {$C$};
%\fill [color=black] (0,1.76) circle (1.5pt);
%\draw[color=black] (0.47,2.54) node {$D$};
%\draw[color=black] (2.2,0.75) node {$E$};
%\draw[color=black] (0.41,-1.04) node {$F$};
%\draw[color=black] (-1.32,0.75) node {$G$};

\fill  (0,0) circle (1pt);

\foreach \x in {90, 95,...,227} {
% \draw (17.25* {sin( \x )} ,17.25*{cos(\x ) }  )   circle (1cm )  ;
 \draw  (   {18 * sin( \x )} , {18 * cos(\x ) }  )   circle (1cm )  ;
}

\foreach \x in {180, 186,...,367} {
 \draw  (   {15.5 * sin( \x )} , {15.5 * cos(\x ) }  )   circle (1cm )  ;
}

\foreach \x in {270, 278,...,384} {
 \draw  (   {13 * sin( \x )} , {13 * cos(\x ) }  )   circle (1cm )  ;
}

\foreach \x in {0, 10,...,50} {
 \draw  (   {10.5 * sin( \x )} , {10.5 * cos(\x ) }  )   circle (1cm )  ;
}

\foreach \x in { 3,4.5,...,17} {
 \draw  (  {\x}, 0)   circle (1cm )  ;
}

\foreach \x in {3,4.5,...,15.5} {
 \draw  (  0, -{\x} )   circle (1cm )  ;
}

\foreach \x in {3,4.5,...,13} {
 \draw  (  -{\x},0 )   circle (1cm )  ;
}

\foreach \x in {3,4.5,...,10.5} {
 \draw  (  0, {\x}  )   circle (1cm )  ;
}

\end{scriptsize}
\end{tikzpicture}
\label{fvsVG}
\caption{A vertex gadget for the reduction from FVS to All-Cells-Connection. The red disk is centered at the corresponding vertex in the grid embedding. The dashed circle has radius $s = r(\sqrt{36n^4+1}-1)$ as computed in Section \ref{secP1}.}
\end{figure}

\section{Hardness of Multiterminal Cut on Unit Disk Graphs}
\label{MCSec}

In this section we prove Theorem \ref{mainThm2} by reducing a restricted version of the Multiterminal Cut problem, described in Theorem \ref{spec} to the Unit Disk Multiterminal Cut Problem.

\begin{theorem}  
The Multiterminal Cut Problem remains $\mathsf{NP}$-complete on unit disk graphs if $k$ is not fixed.
\label{mainThm2}
\end{theorem}

\begin{theorem}  \cite{multicut}
The Multiterminal Cut Problem is  $\mathsf{NP}$-complete for planar graphs if $k$ is not fixed, even if the edge weights are in $\{1, \ldots, 5\}$ and the maximum vertex degree is $3$.
\label{spec}
\end{theorem}

For the reduction we use an edge gadget as shown in Figure \ref{edgeG2} which has the same dimension as the edge gadget of Definition \ref{edgeGdEf}. Thus, any distance bounds computed in the last section also apply here. 
We replace the single paths in the edge gadgets of Definition \ref{edgeGdEf} by $w$ many paths connecting each of the paths in the cabin to $w$ slightly perturbed copies of a single disk touching the boundary. Figure \ref{edgeG2} shows an example of an edge gadget for an edge of weight $5$. By this construction we achieve that the minimum number of edges in the unit disk graph which have to be removed to disconnect any $uv$-path inside the edge gadget for $\{u,v\}$ is $w$. For a vertex $v$ the vertex gadget shown in Figure \ref{vertexG2} consists of 16 slightly perturbed copies of a cycle of $C_V = \lceil \pi s/r \rceil$ disks of radius $r$ which are arranged on a circle of radius $s$. We denote the arrangement of those 16 copies by $\sigma$. Furthermore, the centroid disk representing $v$ is connected to  $\sigma$ by 16 slightly perturbed copies of a path of unit disks which we denote by $\gamma$. Since the shape of the edge gadgets and the radius of the vertex gadgets are the same as in the last section, it holds that all edge gadgets are disjoint and no edge gadget intersects any vertex gadget other than the ones of its end vertices, where the $w$ copies of its last disk intersect between one and three 16-disk clusters of the vertex gadget.\\
For the reduction we take an instance $I=(G, S)$ of the restricted Multiterminal Cut Problem described in Theorem \ref{spec} and embed $G = (V,E)$ crossing free into an $n \times n$ grid, replace each edge  by an edge gadget and each vertex by a vertex gadget as described above, thus obtaining the embedded unit disk graph $G'$. Furthermore, we let $S'$ correspond to the vertices in $S$. In order to proof Theorem \ref{mainThm2} we need to show how to retrieve a solution for $I$ from the solution of the Unit Disk Multiterminal Cut Problem $(G', S')$ in polynomial time.
\begin{lemma}
An instance $I=(G, S)$ of the restricted Multiterminal Cut Problem described in Theorem \ref{spec} has a solution of size $k$ if and only if $I'=(G', S')$ of the  Unit Disk Multiterminal Cut Problem has a solution of size $k$, where $I'$ is built out of $I$ using the construction described above.
\end{lemma}
\begin{proof}
Given $I=(G, S)$ we create the embedded unit disk graph $G'$ using the reduction described above. Note that an optimal solution of the restricted Multiterminal Cut Problem does not remove any edge of the vertex gadget for any vertex $v$, since removing less than 16 edges from the vertex gadget does not disconnect the unit disk graph. On the other hand, removing at most 5 edges inside each of the at most three adjacent edge gadgets will disconnected  $c(v)$ from any other vertex in $S'$. Furthermore, for an edge $e$ with weight $w(e)$, removing $w(e)$ edges in its edge gadget disconnects the two endpoints of the gadget, while removing fewer than $w$ edges still keeps the two endpoints connected and thus, the lemma follows.
\end{proof}
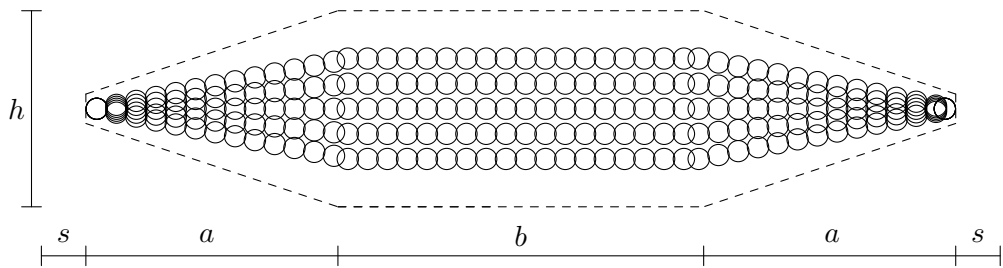
\begin{figure}
\begin{tikzpicture}[line cap=round,line join=round,>=triangle 45,x=1.3cm,y=1.3cm]
\clip(-0.38,-2) rectangle (11.37,2);
\draw [dashed ](3,1)-- (6.7,1);
\draw [dashed ](3,-1)-- (6.7,-1);
\draw [dashed ](3,-1)-- (4.54,-1);
\draw [dashed ](0.45,0.15)-- (0.45,-0.15);
\draw [dashed ] (0.45,0.15)-- (3,1);
\draw [dashed ] (0.45,-0.15)-- (3,-1);
\draw [dashed ] (9.25,0.15)-- (9.25,-0.15);
\draw [dashed ] (9.25,0.15)-- (6.7,1);
\draw [dashed ] (9.25,-0.15)-- (6.7,-1);

%scale
\draw (0,-1.5)-- (9.7,-1.5);

\draw (0.225,-1.3) node {$s$};
\draw (0,-1.40)-- (0,-1.6);
\draw (1.675,-1.3) node {$a$};
\draw (0.45,-1.40)-- (0.45,-1.6);
\draw (3,-1.40)-- (3,-1.6);
\draw (4.85,-1.3) node {$b$};
\draw (6.7,-1.4)-- (6.7,-1.6);
\draw (8,-1.3) node {$a$};
\draw (9.25,-1.4)-- (9.25,-1.6);
\draw (9.475,-1.3) node {$s$};
\draw (9.7,-1.4)-- (9.7,-1.6);

\draw (-0.1,-1)-- (-0.1,1);
\draw (-0.2,1)-- (0,1);
\draw (-0.2,-1)-- (0,-1);
\draw (-0.25, 0) node {$h$};

\foreach \x in {0,0.2,...,2.56} {
      \draw (\x +0.56,0) circle (4pt);
      \draw (\x +0.56,\x*0.1) circle (4pt);
      \draw (\x +0.56,\x*0.2) circle (4pt);
      \draw (\x +0.56,\x*-0.2) circle (4pt);
      \draw (\x +0.56,\x*-0.1) circle (4pt);
}

\foreach \x in {0,0.2,...,2.56} {
      \draw (\x +6.65,0) circle (4pt);
      \draw (\x +6.65,0.256-\x*0.1) circle (4pt);
      \draw (\x +6.65,0.512-\x*0.2) circle (4pt);
      \draw (\x +6.65,-0.256 + \x*0.1) circle (4pt);
      \draw (\x +06.65,-0.512+\x*0.2) circle (4pt);
}

%\foreach \x in {3.15,3.45,...,6.5} {
\foreach \x in {3.1,3.3,...,6.5} {
\foreach \y in {0,1,...,4} {
      \draw (\x ,0.512- \y*0.256) circle (4pt);
}
}

 \draw (9.14 ,0) circle (4pt);
 \draw (9.14 ,0) circle (4pt);
 \draw (9.14 ,0) circle (4pt);
 \draw (9.14 ,0) circle (4pt);
 \draw (9.14 ,0) circle (4pt);

\end{tikzpicture}
\caption{An example of an Edge Gadget for an edge of weight $5$ in the proof of Theorem \ref{mainThm2}.}
\label{edgeG2}
\end{figure}

\begin{figure}
\begin{tikzpicture}[line cap=round,line join=round,>=triangle 45,x=1.3cm,y=1.3cm]
\clip(-0.38,-1.2) rectangle (11.37,2);

\foreach \x in {0, 10,...,350} {
 \draw ({4+ sin(\x)} ,{cos(\x)} ) circle (4pt);
}

\foreach \x in {0.1,0.2,...,1} {
 \draw (4+\x,0) circle (4pt);
}
\draw  (4,-0.2) node[anchor=east] {$c(v)$};
\fill  (4,0) circle (1pt);

\draw  (4.65,0.25) node[anchor=east] {$\gamma$};
\draw  (3.5,1.025) node[anchor=east] {$\sigma$};
 \draw[color=red]  (4,0) circle  (4pt);

\end{tikzpicture}
\caption{Vertex gadget for vertex $v$, where each black disk represents 16 copies of a single disk, and the red disk is the centroid disk $c(v)$.}
\label{vertexG2}
\end{figure}
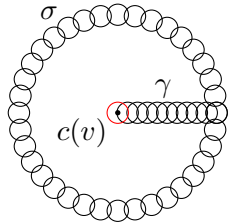

\section{Acknowledgments}

This work started when the second author visited Rolf Klein and his research group. He would like to thank them for the friendly atmosphere they provided. Furthermore, he would like to thank Peter Bra{\ss}  for various helpful comments  and Peter Terlecky for mentioning the problem to him.

\bibliographystyle{plain}

\end{document}